\documentclass[journal,twoside,web]{ieeecolor}
\usepackage{generic}
\usepackage{cite}
\usepackage{amsmath,amssymb,amsfonts}
\usepackage{algorithmic}
\usepackage{graphicx}
\usepackage{algorithm,algorithmic}
\usepackage{hyperref}
\hypersetup{hidelinks=true}
\usepackage{textcomp}
\usepackage{epstopdf}
\newtheorem{assum}{Assumption}
\newtheorem{thm}{Theorem}
\newtheorem{prob}{Problem}
\newtheorem{lem}{Lemma}
\newtheorem{rem}{Remark}

\def\BibTeX{{\rm B\kern-.05em{\sc i\kern-.025em b}\kern-.08em
    T\kern-.1667em\lower.7ex\hbox{E}\kern-.125emX}}
\begin{document}
\title{Localized Data-driven Consensus Control}
\author{Zeze Chang, Junjie Jiao, \IEEEmembership{Member, IEEE}, and Zhongkui Li, \IEEEmembership{Senior Member, IEEE}
\thanks{
This work was supported in part by the National Natural Science Foundation of China under grants U2241214, T2121002, and 62373008. }
\thanks{ 
Zeze Chang and Zhongkui Li are with the State Key Laboratory for Turbulence and Complex Systems, Department of Mechanics and Engineering Science,
College of Engineering, Peking University, Beijing 100871, China (e-mail: changzeze@stu.pku.edu.cn; zhongkli@pku.edu.cn).
}
\thanks{Junjie Jiao is with Chair of Information-oriented Control, Technical University of Munich, Munich 80333, Germany (email:j.jiao@tum.de).}
}

\maketitle

\begin{abstract}
This paper considers a localized data-driven consensus problem for leader-follower multi-agent systems with unknown discrete-time agent dynamics, where each follower computes its local control gain using only their locally collected state and input data.  Both noiseless and noisy data-driven consensus protocols are presented, which can handle the challenge of the heterogeneity in control gains caused by the localized data sampling and achieve leader-follower consensus. 
The design of these data-driven consensus protocols involves low-dimensional linear matrix inequalities. In addition, the results are extended to the case where only the leader's data are collected and exploited. The effectiveness of the proposed methods is illustrated via simulation examples.
\end{abstract}


\section{Introduction} \label{sec 1}
Over the past decades, there has been a growing interest in data-driven control due to its advantages of not requiring precise system models that often contain redundant parameters and are difficult to be identified accurately through experiments. Various methods have been presented to solve data-driven control problems, including model-free adaptive control \cite{Lei2020data,Li2022Distributed}, iterative learning control \cite{Hui2023Data,Zheng2023Data}, reinforcement learning \cite{Wang2021Data,Dong2023Data}, and Willems' fundamental lemma \cite{Willems2005,Waarde2020Willems,Pan2023On}, the last of which does not rely on parametric system identification and has rigorous stability analysis and thereby exhibits a promising prospect for addressing widespread black-box systems.

Quite a few studies have emerged on data-driven control for linear systems utilizing the paradigm proposed in \cite{Willems2005}. In \cite{DePersis2020}, several problems including stabilization, optimality, and robust control are tackled for general linear systems, where the controllers are all devised based on sampled data. The control inputs in \cite{DePersis2020}, nevertheless, are supposed to be persistently exciting (PE) during the data-sampling process such that system matrices can be explicitly derived from the sampled data. An informativity approach is introduced in \cite{Waarde2020}, providing rigorous analysis to determine the necessity of the PE condition, and the data-driven algebraic regulator problem is investigated in \cite{Trentelman2022} by adopting this informativity approach. Data-driven model reduction and data-driven control are studied in \cite{Monshizadeh2020} to further relax some rank requirements of the collected data. 

Due to the existence of external interference, sampled data are very likely to be affected by noises \cite{Furieri2023Near}. As a result, it makes great sense to conduct system analysis and controller design directly from noise-corrupted data. To solve the noisy data-driven control problems, many robust control tools are utilized, including linear fractional transformations \cite{Berberich2020,Sinha2021LFT,Berberich2023Combining}, S-procedure \cite{Waarde2022,Bisoffi2021,Waarde2022Dissipativity,Guo2022}, system level synthesis \cite{Anderson2019System,Xue2020Data}, Petersen's lemma \cite{Bisoffi2022data}, and Finsler's lemma \cite{Waarde2021Finsler}. 
The quadratic regulator problem is studied in \cite{DePersis2021Low,Dörfler2023On,Dai2023Data} for linear systems with noisy data, in which semi-definite program techniques are utilized and sufficient conditions are provided to return a stabilizing controller with guaranteed relative errors.

More recently, efforts have been made to extend the above results for data-driven control of single systems to the setting of network systems, of which the structural constraint induced by the network topology and the diverse data locally collected by each agent impose inherent challenges. 
In \cite{Jiao2021}, data-based output synchronization is studied for heterogeneous leader-follower multi-agent systems, where the output regulation equations are solved solely by data.  However, the noises in the agent dynamics are assumed to be known exactly, which is unrealistic. Additionally,  in \cite{Allibhoy2021}, a distributed predictive control scheme is devised based on sampled data to stabilize coupled network systems. More related results can be found in \cite{Baggio2021data} and \cite{Liu2020Predictive}. Furthermore, interesting studies in \cite{Zhang2023} and  \cite{Wang2022} present data-driven leader-follower consensus and event-triggered consensus algorithms for multi-agent systems, respectively. Nonetheless, in these works, data are collected only at one agent, based on which a local control gain is computed and shared across all agents. This is essentially a {\em centralized} design.

Motivated by the above discussions, in the present paper we deal with distributed data-driven control of homogeneous leader-follower multi-agent systems from a new perspective. 
In particular, we consider the case where each follower collects its own local data and uses such locally collected noisy/noiseless data to design its own local control gain. This, different from those centralized design methods using only one agent's data as in \cite{Zhang2023,Wang2022}, is called {\em localized} data-driven control.
Note that in this case, the local control gains of the followers are, in general, different from each other,   
introducing heterogeneity, which makes the classic decomposition methods for the homogeneous case not applicable directly.
To overcome this difficulty,  we propose a novel distributed data-driven consensus protocol, by additionally synchronizing these local gains, which is shown to achieve leader-follower consensus.

We consider both cases of noiseless data and noisy data. Specifically, we devise noiseless data-driven consensus protocols based on the Riccati-based consensus region and develop noisy data-driven protocols leveraging the S-procedure technique and the informativity approach. Additionally, we provide another type of data-driven consensus algorithm, where only the leader's data are exploited.  Rigorous analyses are provided to show the convergence of the closed-loop network system, without requiring the system matrices to be known or the PE condition to hold. Simulation examples are also provided to verify the effectiveness of the proposed algorithms.

The main contributions of this paper are at least three-fold: 1)  The localized data-driven consensus algorithms presented in this paper, circumventing the centralized design approach in \cite{Zhang2023,Wang2022} which necessitates an identical data-based gain for all agents, determine local gains for each agent in a coordinated way using locally collected data, which is consistent with the essential nature of distributed control. 
2) Data-driven consensus is investigated under a much relaxed assumption on the noisy data. Contrary to the previous work \cite{Jiao2021} which demands the exact prior knowledge of the noise signals, in the current paper the noise matrix can be unknown and is only required to satisfy certain bounded constraints. 
3) The data-driven consensus protocols are designed by solving low-dimensional linear matrix inequalities (LMIs) in this paper. By contrast,  the LMI conditions obtained in \cite{Wang2022}, proportional to the scale of the network, are generally of high dimension.

The organization of this paper is as follows.  In Section \ref{sec 3}, the problem formulation and the noiseless data-driven consensus control protocols are proposed.  Section \ref{sec 4} extends the analysis to the data-based consensus algorithm with noise-corrupted data. Section \ref{sec Centralized} provides another data-driven control architecture, where only the leader's data are exploited. Simulation results are demonstrated in Section \ref{sec 5} to illustrate the effectiveness of the proposed methods.  Section \ref{sec 6} concludes this paper.

{\em Notations} : 
$\mathbb{R}^{n\times m}$ denotes the set of $n\times m$ real matrices. $A\otimes B$ represents the Kronecker product of matrices $A$ and $B$.
$\sigma_{\rm {max}}(A)$ and $\lambda_{\rm {max}}(A)$ denote the maximum singular value and the largest eigenvalue of matrix $A$, respectively. $C(a,b)$ denotes an open circle in the complex plane $\mathbb{C}$ with radius $b$ centered at $a\in\mathbb{C}$. The corresponding closed circle is denoted as $\bar{C}(a,b)$. $[*]^T$ denotes the elements that can be inferred by the matrix symmetry.

\section{Localized Data-driven Control with Noise-free Data}\label{sec 3}

In this section, we are concerned about localized data-driven consensus of multi-agent systems using noise-free data. We will first formulate the problem in Subsection \ref{subsection pf} and then provide the method for designing data-driven consensus protocols in Subsection \ref{subsection 3.2}.

\subsection{Problem formulation} \label{subsection pf}

Consider a group of $N+1$ homogeneous agents consisting of $N$ followers and a leader, with discrete-time general linear dynamics characterized by
\begin{equation} \label{equ:1}
x_i(t+1)=Ax_i(t)+B u_i(t), \quad i=0,1,\cdots,N.
\end{equation}
where $x_i\in\mathbb{R}^n$ and $u_i\in\mathbb{R}^p$ denote the state and the control input of the $i$th agent, respectively.
The matrices $A$ and $B$ are of compatible dimensions, but {\em unknown}. Assume that the pair $(A,B)$ is controllable and the control input of the leader $u_0=0$.

The network topology among the $N+1$ agents is represented by a graph $\mathcal{G}=\left\{\mathcal{V},\mathcal{E}\right\}$, where $\mathcal{V}$ denotes the set of nodes and $\mathcal{E} \in \mathcal{V}\times \mathcal{V}$ denotes the set of edges. For an edge $(i,j)$, agent $j$ can have access to information from agent~$i$. 
A directed path from node $i_1$ to node $i_n$ is a sequence of ordered edges of the form $(i_m,i_{m+1})$, $m=1,\cdots,n-1$. A graph contains a directed spanning tree if there exists a node called root such that there exists a directed path from this node to every other node. 
The adjacent matrix $\mathcal{A}$ of graph $\mathcal{G}$ is defined as $a_{ii}=0$, $a_{ij}>0$ if $(j,i)\in\mathcal{E}$, $a_{ij}=0$ otherwise. The Laplacian matrix is defined as $L=\mathcal{D}-\mathcal{A}$, where $\mathcal{D}=\rm{diag}(d_0,\cdots,d_N)$ is the degree matrix with $d_i=\Sigma_{j=0}^{N} a_{ij}$. 
The Laplacian matrix can be partitioned as
$L=\begin{bmatrix}
0 & 0\\
L_{fl} & L_{ff}
\end{bmatrix}.$
Clearly, $0$ is the eigenvalue of the Laplacian matrix $L$ with an associated eigenvector $\mathbf{1}$. Furthermore, $0$ is a simple eigenvalue of $L$ if $\mathcal{G}$ has a directed spanning tree \cite{Agaev2005On}.

\begin{assum} \label{asp 1}
	The graph $\mathcal{G}$  contains a directed spanning tree with the leader as the root and the subgraph among followers is undirected.
\end{assum}


Since the system matrices $A$  and $B$ of \eqref{equ:1}  are unknown, we assume that we have access to system data in order to design the consensus protocol. In particular, we consider the case that each follower collects its own data samples to design its local control gain.
More specifically, for the $i$th follower, we collect state and input data on $T$ finite sequences and construct the following data matrices:
\begin{equation} \label{equ:noiseless data}
\begin{split}
U_{i-}=&\begin{bmatrix}
u_i(0) & u_i(1) & \cdots & u_i(T-1)
\end{bmatrix}, \\
X_i=&\begin{bmatrix}
x_i(0) & x_i(1) & \cdots & x_i(T)
\end{bmatrix},i=1,\cdots,N.
\end{split}
\end{equation}
Next, define
\begin{equation} \nonumber
\begin{split}
X_{i-}=&\begin{bmatrix}
x_i(0) & x_i(1) & \cdots & x_i(T-1) 
\end{bmatrix}, \\
X_{i+}=&\begin{bmatrix}
x_i(1) & x_i(2) & \cdots & x_i(T)
\end{bmatrix}.
\end{split}
\end{equation}
According to (\ref{equ:1}), clearly we have $X_{i+}=AX_{i-}+BU_{i-}$ .

\begin{assum}\label{asp 3}
	The data matrix $\begin{bmatrix}
	U_{i-} \\
	X_{i-}
	\end{bmatrix}$ has full row rank for $i=1,\cdots,N$.
\end{assum}

Note that Assumption \ref{asp 3} can be easily satisfied by exploiting adequate data to make $X_{i-}$ and $U_{i-}$ wide enough. Obviously, under Assumption \ref{asp 3}, $X_{i-}$ has full row rank and has a right inverse denoted as $X_{i-}^{\dagger}$.

The goal of this section is to devise a distributed protocol for (\ref{equ:1}) based on the collected data (\ref{equ:noiseless data}) such that {\em leader-follower consensus} is achieved, i.e., $\lim\limits_{t\to \infty} [x_i(t)-x_0(t)]=0$ for $i=1,\cdots, N$. 

Note that in general, the collected data for each follower are distinct, resulting in different locally designed gains. In this case, the classical model-based distributed controller in \cite{Movric2013Synchronization,Li2014} is not readily applicable, since it employs an identical gain. To overcome this difficulty, we propose the following distributed control law:
\begin{subequations} \label{equ:dis_control_law}
	\begin{align}	u_i(t)=&K_i(t)\sum_{j=0}^N \frac{a_{ij}}{1+d_i} 
	\bigg(x_i(t)-x_j(t)\bigg),\label{equ:7a} \\
	K_i(t+1)=&K_i(t)+\mathcal{O}\sum_{j=1}^N \frac{a_{ij}}{1+z}\bigg(K_i(t)-K_j(t)\bigg),\label{equ:7b}
	\end{align}
\end{subequations} 
for $i=1,\cdots,N$,
where $a_{ij}$ is the element of the adjacent matrix $\mathcal{A}$, $z=\rm{max}_i(\sum_{j=1}^N a_{ij})$, and $\mathcal{O}$ is the feedback gain matrix. 
Similar to \cite{Movric2013Synchronization}, we refer to matrix $\bar{L}=(I+\mathcal{D}_{ff})^{-1}L_{ff}$ as the weighted graph matrix, where $\mathcal{D}_{ff}$ is the degree matrix of the followers. It is worth mentioning that all the eigenvalues of $\bar{L}$ satisfy that $\lambda_k\in\bar{C}(1,1)$, $k=1,\cdots,N$ for any graph \cite{Movric2013Synchronization}.

The problem we want to address in this section is then described as follows.
\begin{prob}\label{prob 1}
	Design gain matrices $K_i(0)$ and  $\mathcal{O}$ for the followers using collected data such that the associated protocol \eqref{equ:dis_control_law} achieves leader-follower consensus for the agents in \eqref{equ:1}.
\end{prob}

\subsection{Data-driven control design}\label{subsection 3.2}

The design of the data-driven protocol (\ref{equ:dis_control_law}) contains the following steps: 
1) Compute the initial feedback gain matrix $K_i(0)$ in  (\ref{equ:dis_control_law}) for each follower directly from the noise-free data; 2) Calculate the data-based solution to an algebraic Riccati equation, which is important in determining the consensus region \cite{Li2010Consensus};
3) Utilize the data (\ref{equ:noiseless data}) sampled from each follower to establish the consensus region;
4) Show that the protocol (\ref{equ:dis_control_law}) with parameters obtained in the above steps can achieve leader-follower consensus.

We first present the results for obtaining the initial gain matrix $K_i(0)$.

\begin{thm} \label{thm1}
	Let Assumptions \ref{asp 1}-\ref{asp 3} hold. Suppose that there exists $\Gamma_i$ such that
	\begin{equation} \label{equ:K0}
	\begin{split}
	& \mathop{\rm {min}}\limits_{\Gamma_i} \,\, {\rm {Trace}}(Q_iX_{i-}\Gamma_i)\\
	&\,\,s.t., \begin{bmatrix}
	X_{i-}\Gamma_i-I_n & X_{i+}\Gamma_i \\
	\Gamma_i^TX_{i+}^T & X_{i-}\Gamma_i  
	\end{bmatrix}\geq 0, \\
	&\quad\,\,\,\,X_{i-}\Gamma_i\geq I_n,\,\,\,\,\,\, i=1,\cdots,N,
	\end{split}
	\end{equation}
	where $Q_i>0$ is a constant matrix.
	Then, the initial feedback gain matrix $K_i(0)$ of the $i$th agent can be calculated as $K_i(0)=U_{i-}\Gamma_i(X_{i-}\Gamma_i)^{-1}$. 
\end{thm}

\begin{proof}
It follows from (\ref{equ:1}) and (\ref{equ:noiseless data}) that  $X_{i+}=AX_{i-}+BU_{i-}$.
In light of the classic data-driven results in \cite{DePersis2020}, we design the initial controller as $K_i(0)=U_{i-}X_{i-}^{\dagger}$. It then follows that $A+BK_i(0)=X_{i+}X_{i-}^{\dagger}$.  

Despite the fact that $A,B$ are unknown, we need to guarantee that such $K_i(0)$ adheres to the structural constraint of $K_i(0)=-(B^TP_iB)^{-1}B^TP_iA$, where $P_i>0$ is the solution to the following algebraic Riccati equation (ARE):
\begin{equation} \label{equ:mare}
P_i=A^TP_iA-A^TP_iB(B^TP_iB)^{-1}B^TP_iA+Q_i.
\end{equation}

Recalling the well-known result in \cite{Balakrishnan1995}, the initial feedback matrix $K_i(0)$ with the aforementioned structural constraint can be obtained via the following dual optimization problem:
\begin{equation}\label{equ:nonconvex}
\begin{split}
& \mathop{\rm {min}}_{\Xi_i,K_i(0)} \,\,\,{\rm {Trace}} (Q_i\Xi_i) \\
&s.t.\,\,(A+BK_i(0))\Xi_i(A+BK_i(0))^T-\Xi_i \leq -I_n, \\
&\,\,\,\,\quad\Xi_i\geq I_n.
\end{split}
\end{equation} 
Let $\Gamma_i=X_{i-}^{\dagger}\Xi_i$. Then, we can proceed analogously to the results in \cite{DePersis2020} and obtain (\ref{equ:K0}) mutatis mutandis, where $K_i(0)=U_{i-}\Gamma_i(X_{i-}\Gamma_i)^{-1}$. This completes the proof.
\end{proof}

Note that Theorem \ref{thm1} only provides the initial feedback gain matrix $K_i(0)$, rather than the unique solution $P_i$ to (\ref{equ:mare}), which plays a vital role in computing the consensus region \cite{Li2014}. In the following result, we will develop an approach to calculate $P_i$.

\begin{thm} \label{thm2}
	The solution $P_i$ to the ARE (\ref{equ:mare}) can be obtained via the following optimization problem:
	\begin{equation} \label{equ:P_noise}
	\begin{split}
	& \mathop{\rm{max} }_{P_i}\,\,\, {\rm {Trace}} (P_i) \\
	&s.t. \,\,\, P_i=P_i^T\geq 0 \,\,\, \\
	&\,\,\, \,\,\,\,\, \,\,\,[X_{i+}\Gamma_i(X_{i-}\Gamma_i)^{-1}]^TP_i[X_{i+}\Gamma_i(X_{i-}\Gamma_i)^{-1}]\\
	&\,\,\,\,\,\,\,\,\,\,-P_i+Q_i\geq 0,
	\end{split}
	\end{equation}
	where $i=1,\cdots,N$.
\end{thm}

\begin{proof}
{ Notice that the solution of  (\ref{equ:mare}) is also the  solution of following optimization problem} \cite{Ran1988}:
\begin{equation} \label{equ:model_based P}
\begin{split}
& \mathop{\rm {max}}_{P_i}\,\,\, {\rm {Trace}} (P_i) \\
& {s.t.} \,\,\, (A+BK_i(0))^TP_i(A+BK_i(0))-P_i\\
& \,\,\, \,\,\, \,\,\,+Q_i\geq 0 \\
&\,\,\,\,\,\,\,\,\,\,\,\, P_i=P_i^T\geq 0.
\end{split}
\end{equation}

Then, substituting  $A+BK_i(0)=X_{i+}\Gamma_i(X_{i-}\Gamma_i)^{-1}$ and $K_i(0)=U_{i-}\Gamma_i(X_{i-}\Gamma_i)^{-1}$ into (\ref{equ:model_based P}), with $\Gamma_i$ obtained in Theorem \ref{thm1}, yields the data-based convex program \end{proof}

\begin{lem}~\cite{Li2014} \label{lem 1}
	Suppose that Assumption \ref{asp 1} holds and $\mathcal{O}=-(\Lambda+\tilde{R})^{-1}\Lambda$, where $\Lambda$ is the solution to the following modified algebraic Riccati equation (MARE) \cite{Schenato2007,Sinopoli2004}:
	\begin{equation}\nonumber
	\Lambda=\mathcal{T}^T\Lambda\mathcal{T}-(1-\delta^2)\mathcal{T}^T\Lambda\mathcal{T}(\mathcal{T}^T\Lambda\mathcal{T}+\tilde{R})^{-1}\mathcal{T}^T\Lambda \mathcal{T}+\tilde{Q},
	\end{equation}
	in which $\mathcal{T}=I_p$,  $\tilde{R}>0$, $\tilde{Q}>0$, and $\lambda_{\rm {max}}(D_{ff})<\delta<1$, $D_{ff}\in\mathbb{R}^{N\times N}$ is a row-stochastic matrix defined as $d_{ij}=a_{ij}/(1+z)$  and $d_{ii}=1-\sum_{j=1}^{N}a_{ij}/(1+z)$, and $z$ is defined as in (\ref{equ:dis_control_law}).
	If all the non-one eigenvalues of $D_{ff}$ are located in $\Gamma_{\leq \sigma}$, the disk of radius $\delta$ centered at the origin, 
	then the feedback matrices $K_i(t)$, $i=1,\cdots, N$, in (\ref{equ:7b}),
	converge exponentially to $\frac{1}{N}\sum_{i=1}^{N}K_i(0)$.
\end{lem}

The consensus region is designed in the following result.

\begin{lem} \label{lem 2}
	Suppose that Assumptions \ref{asp 1}-\ref{asp 3} hold and the input matrix $B$ in (\ref{equ:1}) is invertible.
	Let $G_i$ be one solution to the following equation:
	\begin{equation} \label{equ:L}
	\begin{split}
	\begin{bmatrix}
	K_i(0)\\
	0
	\end{bmatrix}=
	\begin{bmatrix}
	U_{i-} \\
	X_{i-}
	\end{bmatrix}G_i, ~i=1,\cdots,N,
	\end{split}
	\end{equation}
	where $K_i(0)$ is calculated as in Theorem \ref{thm1}.
	If all the eigenvalues of $\bar{L}$ are located in the consensus region: 
	$$\mathcal{Y}=\left\{\eta\,|\,|\eta-1|^2<\frac{1}{\sigma_{\rm max}(\mathcal{F}^{-1/2}\mathcal{R}\mathcal{F}^{-1/2})} \right\},$$
	with $\mathcal{R}=\sum_{i=1}^N[G_i^TX_{i+}^TP_iX_{i+}G_i+\mathcal{W}_i^T(P-P_i)\mathcal{W}_i]$,  $\mathcal{F}=\sum_{i=1}^N[Q_i+(P-P_i)]$, $\mathcal{W}_i=(X_{i+}\mathcal{U}_i-X_{i+}G_i)$, $\mathcal{U}_i=\Gamma_i(X_{i-}\Gamma_i)^{-1}$,   $P=\sum_{i=1}^{N}P_i $, and $P_i$ computed as in Theorem \ref{thm2}, then $K_0$ renders $I_N\otimes A+\bar{L}\otimes BK_0$ Schur stable, where $K_0=\frac{1}{N}\sum_{i=1}^{N}K_i(0)$.
\end{lem}


\begin{proof} 
Motivated by \cite{Movric2013Synchronization,Li2014}, the consensus region $\mathcal{Y}$ associated with $K_0$ can be derived as follows:
\begin{equation} \label{equ:region1}
\begin{split}
&(A+\eta BK_0)^HP(A+\eta BK_0)-P \\
=& A^TPA+\eta A^TPBK_0+\eta^H K_0^TB^TPA\\
&+|\eta|^2 K_0^TB^TPBK_0-P\\
\leq&A^TPA+ \frac{\eta}{N}A^TPB\sum_{i=1}^{N}K_i(0)+\frac{\eta^H}{N}\sum_{i=1}^{N}K_i^T(0)\\
&\times B^TPA+\frac{|\eta|^2}{N}\sum_{i=1}^{N}K_i^T(0)B^TPBK_i(0)-P,
\end{split}
\end{equation}
where the last inequality is obtained by using the  Young's inequality \cite{Boyd1994Linear}.

Substituting $P=\sum_{i=1}^{N}P_i$ into (\ref{equ:region1}) yields
\begin{equation} \label{equ: lyapunov}
\begin{split}
&(A+\eta BK_0)^HP(A+\eta BK_0)-P \\
\leq&\frac{1}{N}\sum_{i=1}^{N}[A^TP_iA-2\rm{Re}(\eta) K_i(0)^TB^TP_iBK_i(0)\\
&+|\eta|^2 K_i(0)^TB^T P_iBK_i(0)-P_i]+\frac{1}{N}\sum_{i=1}^{N}[A^T\\
&\times(P-P_i)A-\eta A^T(P-P_i)B(B^TP_iB)^{-1}B^T\\
&\times P_iA-\eta^H A^T P_iB(B^TP_iB)^{-1}B^T(P-P_i)A\\
&+|\eta|^2K_i(0)^TB^T(P-P_i)BK_i(0)-(P-P_i)].\\
\end{split}
\end{equation}
Note that
\begin{equation} \label{equ: lyapunov_1}
\begin{split}
&\frac{1}{N}\sum_{i=1}^{N}[A^TP_iA-2\rm{Re}(\eta) K_i(0)^TB^TP_iBK_i(0)\\
&+|\eta|^2
K_i^T(0) B^TP_iBK_i(0)-P_i]\\
=&\frac{1}{N}\sum_{i=1}^{N}[|\eta-1|^2K_i(0)^TB^TP_iBK_i(0)-Q_i],
\end{split}
\end{equation}
and
\begin{equation} \label{equ: lyapunov_2}
\begin{split}
&\frac{1}{N}\sum_{i=1}^{N}[A^T(P-P_i)A
-\eta A^T(P-P_i)B(B^TP_iB)^{-1}\\
&\times B^TP_iA-\eta^H A^T
P_iB(B^TP_iB)^{-1}B^T(P-P_i)A\\
&+|\eta|^2K_i(0)^TB^T(P-P_i)BK_i(0)-(P-P_i)]\\
=& \frac{1}{N}\sum_{i=1}^{N}[A^T(P-P_i)A
-2\rm{Re}(\eta) A^T(P-P_i)A\\
&+|\eta|^2A^T(P-P_i)A
-(P-P_i)]\\
=&\frac{1}{N}\sum_{i=1}^{N}[|\eta -1|^2
A^T(P-P_i)A-(P-P_i)],
\end{split}
\end{equation}
where (\ref{equ: lyapunov_1}) is obtained according to (\ref{equ:mare}) and the first equality in (\ref{equ: lyapunov_2}) is derived using the fact that $B$ is invertible. Then, substituting (\ref{equ: lyapunov_1}) and (\ref{equ: lyapunov_2}) into (\ref{equ: lyapunov}) gives
\begin{equation} \label{equ:model_based_region}
\begin{split}
&(A+\eta BK_0)^HP(A+\eta BK_0)-P \\
\leq&\frac{1}{N}\sum_{i=1}^{N}[|\eta-1|^2(K_i(0)^TB^TP_iBK_i(0)+A^T(P-P_i)A)\\
&-(Q_i+(P-P_i))].
\end{split}
\end{equation} 

Note that the matrix $G_i$ in Lemma \ref{lem 2} satisfies the condition  
$\begin{bmatrix}
K_i(0)\\
0
\end{bmatrix}=
\begin{bmatrix}
U_{i-} \\
X_{i-}
\end{bmatrix}G_i$, which is easy to obtain, since $\begin{bmatrix}
U_{i-} \\
X_{i-}
\end{bmatrix}$ has full row rank and $K_{i}(0)$ is given from~(\ref{equ:K0}). It is easy to verify that $BK_i(0)=\begin{bmatrix}
B & A
\end{bmatrix}\begin{bmatrix}
K_i(0)\\
0
\end{bmatrix}=\begin{bmatrix}
B & A
\end{bmatrix}\begin{bmatrix}
U_{i-} \\
X_{i-}
\end{bmatrix}G_i=X_{i+}G_i$. Note that  $A+BK_i(0)=X_{i+}\Gamma_i(X_{i-}\Gamma_i)^{-1}=X_{i+}\mathcal{U}_i$. Then, it follows that $\mathcal{W}_i=X_{i+}\mathcal{U}_i-X_{i+}G_i=A$. 
Substituting  $X_{i+}G_i$ and $\mathcal{W}_i$ into (\ref{equ:model_based_region}), we can obtain that $(A+\eta BK_0)^HP(A+\eta BK_0)-P<0$ for any $\eta\in\mathcal{Y}$, meaning that $I_N\otimes A+\bar{L}\otimes BK_0$ Schur stable, if all the eigenvalues of $\bar{L}$ are located in $\mathcal{Y}$.  \end{proof}

\begin{rem}
	The input matrix $B$ is assumed to be invertible in Lemma \ref{lem 2} to simplify the determination of the data-based consensus region linked to $K_0$. For the case where $B$ is not invertible, {the analysis is in fact similar, while the consensus region will be much more complex.} Let $X_{i+}G_i=\mathcal{X}_i$. In this case the consensus region is  denoted as $\mathcal{Y}=\{\eta\,|\,a|\eta-1|^2+b|\eta|^2+c|\eta|+d\leq 1\}$, where
	\begin{equation} \nonumber
	\begin{split}
	a&=\sigma_{\rm {max}}(\sum_{i=1}^N(\mathcal{F}^{-1/2} \mathcal{X}_{i}^TP_i\mathcal{X}_{i} \mathcal{F}^{-1/2})),\\
	b&=\sigma_{\rm {max}}(\sum_{i=1}^N(\mathcal{F}^{-1/2}  \mathcal{X}_{i}^T(P-P_i)\mathcal{X}_{i} \mathcal{F}^{-1/2})),\\
	c&=\sigma_{\rm {max}}(\sum_{i=1}^N\sum_{j=1,j\neq i}^N(\mathcal{F}^{-1/2}  (\mathcal{X}_{j}^TP_j\mathcal{X}_{i}+\mathcal{X}_{i}^TP_i\mathcal{X}_{j})  \mathcal{F}^{-1/2})),\\
	d&=\sigma_{\rm {max}}(\sum_{i=1}^N(\mathcal{F}^{-1/2} \mathcal{W}_i^T(P-P_i)\mathcal{W}_i \mathcal{F}^{-1/2})),
	\end{split}
	\end{equation}
	in which $\mathcal{F}$, $\mathcal{W}_i$, and $P$ are defined in Lemma \ref{lem 2}. The details are omitted here due to the limited space.
\end{rem}

Before moving on, we provide the following lemma.

\begin{lem}~\cite{Liu2018} \label{lem 3}
	Consider the following system:
	\begin{equation} \nonumber
	x(t+1)=Fx(t)+F_1(t)x(t),
	\end{equation}
	where $x\in \mathbb{R}^n$, $F\in \mathbb{R}^{n\times n}$ is Schur stable, and $F_1(t)$ is well-defined for all $t\in\mathbb{Z}^+$. If $F_1(t)$  exponentially converges to zero, then we have $\lim\limits_{t\to \infty}x(t)=0$.
\end{lem}

We are ready to present the main result of this section.

\begin{thm} \label{thm3}
	Assume that Assumptions \ref{asp 1}-\ref{asp 3} hold and the input matrix $B$ in (\ref{equ:1}) is invertible. { If all the eigenvalues of $\bar{L}$ are located in the consensus region $\mathcal{Y}$ defined as in Lemma \ref{lem 2},} then the data-driven consensus protocol (\ref{equ:dis_control_law}), with the initial feedback matrix $K_i(0)$ calculated in Theorem \ref{thm1} and { $\mathcal{O}$ calculated in Lemma \ref{lem 1},} solves Problem 1.
\end{thm}

\begin{proof}
Let $\tilde{K}_i(t)=K_i(t)-K_0$ and $\tilde{x}_i(t)=x_i(t)-x_0(t)$. According to Lemma \ref{lem 1}, $\tilde{K}_i(t)$ converges to zero. Substituting (\ref{equ:dis_control_law}) into (\ref{equ:1}) yields the following closed-loop network dynamics:
\begin{equation} \label{equ:11}
\begin{split}
x_i(t+1)=Ax_i(t)+&BK_0\sum_{j=0}^{N}\frac{a_{ij}}{1+\sum_{j=0}^N a_{ij}}[\tilde{x}_i(t)-\tilde{x}_j(t)] \\
+B\tilde{K}_i(t)&\sum_{j=0}^{N}\frac{a_{ij}}{1+\sum_{j=0}^N a_{ij}}[\tilde{x}_i(t)-\tilde{x}_j(t)].
\end{split}
\end{equation}
It then follows that
\begin{equation} \label{equ:12}
\begin{split}
\tilde{x}_i(t+1)=A\tilde{x}_i(t)+&BK_0\sum_{j=0}^{N}\frac{a_{ij}}{1+d_i}[\tilde{x}_i(t)-\tilde{x}_j(t)] \\
+B\tilde{K}_i(t)&\sum_{j=0}^{N}\frac{a_{ij}}{1+d_i}[\tilde{x}_i(t)-\tilde{x}_j(t)].
\end{split}
\end{equation}
Let $\tilde{x}(t)=[\tilde{x}_1^T(t),\cdots,\tilde{x}_N^T(t)]^T$ and (\ref{equ:12}) can be rewritten into the following compact form:
\begin{equation} \label{equ:13}
\begin{split}
\tilde{x}(t+1)=&[I_{N}\otimes A+\bar{L}\otimes BK_0]\tilde{x}(t)\\
&+\begin{bmatrix}
\bar{L}_1\otimes B\tilde{K}_1(t)\\
\vdots\\
\bar{L}_{N}\otimes B\tilde{K}_N(t)
\end{bmatrix} \tilde{x}(t),
\end{split}
\end{equation}
where $\bar{L}_i$ represents the $i$th row of $\bar{L}$.
Note that  $\begin{bmatrix}
(\bar{L}_1\otimes B\tilde{K}_1(t))^T&
\cdots&
(\bar{L}_{N}\otimes B\tilde{K}_N(t))^T
\end{bmatrix}^T$ exponentially converges to zero. According to Lemma \ref{lem 2}, the matrix $[I_{N}\otimes A+\bar{L}\otimes BK_0]$ is Schur stable, if all the eigenvalues of $\bar{L}$ are located in the consensus region $\mathcal{Y}$. Therefore,
in view of Lemma \ref{lem 3}, we can derive from (\ref{equ:13}) that $\lim\limits_{t\to \infty}\tilde{x}(t)=0$.  	\end{proof}

\begin{rem}
	Data-driven consensus problem is also studied in previous works \cite{Wang2022,LiData2023,Zhang2023}. However, the control algorithms in these works require an identical data-based feedback gain for all agents, essentially demanding a centralized mechanism to collect data, compute the gain, and assign it to every agent. By contrast, our approach provides a distributed control architecture wherein each follower computes its initial local gain using its own locally sampled data. To tackle the heterogeneity induced by different data-based gains, an interaction mechanism is designed to synchronize the feedback gain $K_i(t)$ in (\ref{equ:7a}).
	Besides, the requirement of system identification in \cite{Zhang2023} is avoided for the data-driven consensus protocol in this section.
\end{rem}

\section{Localized Data-driven Control with Noisy Data}\label{sec 4}

\subsection{Problem formulation} \label{subsection noisy}

In this section, we extend to consider the localized data-driven consensus control problem with noise-corrupted data, leveraging the S-procedure \cite{Polik2007A} and the informativity approach \cite{Waarde2020,Waarde2020data}. In practice, the agent dynamics are often subject to ubiquitous external perturbations such as winds and measurement errors \cite{Dierks2009Output}. Thus, it is imperative to design data-driven consensus control protocols with noise-corrupted data. In this section, we consider the following agents:
\begin{equation} \label{equ:noise_system}
\begin{split}
x_i(t+1)=Ax_i(t)+Bu_i(t)+d_i(t), ~i=0,1,\cdots,N,
\end{split}
\end{equation}
where external disturbances $d_i\in\mathbb{R}^n$ are Lebesgue measurable and bounded. The control input of the leader is still set to be $u_0=0$. 
{ It should be noted that in this section, we consider the presence of process noise signals during the data collection process while maintaining control of the original system (\ref{equ:1}).}
We sample data from both the leader and followers and obtain the following data matrices:
\begin{equation} \label{equ:noise_data}
\begin{split}
X_i=&\begin{bmatrix}
x_i(0) & x_i(1) & \cdots & x_i(T)
\end{bmatrix}\\
U_{i-}=&\begin{bmatrix}
u_i(0) & u_i(1) & \cdots & u_i(T-1)
\end{bmatrix}, \\
\end{split}
\end{equation}
where $i=0,1,\cdots,N$.
Next, define
\begin{equation} \nonumber
\begin{split}
X_{i-}=&\begin{bmatrix}
x_i(0) & x_i(1) & \cdots & x_i(T-1) 
\end{bmatrix}, \\
X_{i+}=&\begin{bmatrix}
x_i(1) & x_i(2) & \cdots & x_i(T)
\end{bmatrix}.\\
D_i=&\begin{bmatrix}
d_i(0) & d_i(1) & \cdots & d_i(T-1)
\end{bmatrix},
\end{split}
\end{equation}
According to (\ref{equ:noise_dynamics}), it is straightforward to note that the data matrices above satisfy the following constraints:
\begin{equation} \label{equ:noise_dynamics}
\begin{split}
X_{i+}=AX_{i-}+BU_{i-}+D_i, \quad i=0,1,\cdots,N.
\end{split}
\end{equation}

We make the following assumption on the additive noise matrix $D_i$, which also appears in several existing results \cite{Waarde2022,Bisoffi2021,Waarde2022Dissipativity}.

\begin{assum} \label{asp 5}
	The noise matrix $D_i$ is unknown and satisfies
	\begin{equation}\label{equ:noise_constraint}
	\begin{split}
	\begin{bmatrix}
	I & D_i
	\end{bmatrix}\begin{bmatrix}
	N_{11} & N_{12} \\
	N_{21} & N_{22}
	\end{bmatrix}\begin{bmatrix}
	I\\
	D_i^T
	\end{bmatrix}\geq 0,~ i=0,1,\cdots,N,
	\end{split}
	\end{equation}
	where { known matrices} $N_{11}>0$, $N_{22}<0$, and $N_{12}^T=N_{21}$ { are of suitable dimensions}. 
\end{assum}

From (\ref{equ:noise_dynamics}) and (\ref{equ:noise_constraint}), we can obtain that 
\begin{equation} \label{equ:noisy_system_identification}
\begin{split}
\begin{bmatrix}
I\\
A^T\\
B^T
\end{bmatrix}^T\begin{bmatrix}
I & X_{i+}\\
0 & -X_{i-}\\
0 & -U_{i-}
\end{bmatrix}\begin{bmatrix}
N_{11} & N_{12} \\
N_{21} & N_{22}
\end{bmatrix}\begin{bmatrix}
I & X_{i+}\\
0 & -X_{i-}\\
0 & -U_{i-}
\end{bmatrix}^T\begin{bmatrix}
I\\
A^T\\
B^T
\end{bmatrix}\geq 0,
\end{split}
\end{equation}
where $i=0,1,\cdots,N$,
implying that all the systems that can generate data (\ref{equ:noise_data}) with the constraint (\ref{equ:noise_constraint}) satisfy (\ref{equ:noisy_system_identification}). Then, as in \cite{Waarde2022}, we define a system set as
$\mathcal{Q}_i=\{(A,B)|(A,B)\; {\rm{satifies}}\; (\ref{equ:noisy_system_identification})\}$. Obviously, the true system $(A,B)\in\mathcal{Q}_i$ for $i=0,1,\cdots,N$.

Instead of (\ref{equ:dis_control_law}), we propose a different data-driven consensus protocol for the agents in (\ref{equ:1}) as follows:
\begin{equation} \label{equ:control_law_noise}
\begin{split}
u_i(t)&=\alpha K_i(t)\sum_{j=0}^{N}a_{ij}\bigg(x_i(t)-x_j(t)\bigg), \\
K_i(t+1)&=K_i(t)+\sum_{j=0}^{N}w_{ij}\bigg(K_j(t)-K_i(t)\bigg),\\ 
\end{split}
\end{equation}
for $i=1,\cdots,N$,
where $\alpha$ is a scalar,  $w_{ij}=\frac{a_{ij}}{1+d_i}$, $w_{ii}=\frac{1}{1+d_i}$.

The problem we intend to address in this section is then described as follows.
\begin{prob}
	Design  gain matrices $K_i(0)$ and  coupling gain $\alpha$ using collected noise-corrupted data (\ref{equ:noise_data}) such that the protocol (\ref{equ:control_law_noise}) achieves leader-follower consensus.
\end{prob}

\subsection{Data-driven control design}\label{subsection noisy_design}

Firstly, we extend the definition of the data informativity \cite{Waarde2020} to the case of the multi-agent system in (\ref{equ:1}).

{\bf Definition 1}~ Suppose that Assumptions \ref{asp 1}-\ref{asp 5} hold. The collected data $(X_i,U_{i-})$ are informative for consensus, if there exists a data-based feedback gain matrix $K_i(0)$ such that $I_N\otimes A+\alpha L_{ff}\otimes BK_i(0)$ is Schur stable for all $(A,B)\in\mathcal{Q}_i$.

Before moving forward, we introduce the following lemma.

\begin{lem}~\cite{Xie1992H}   \label{lem 4}
	There exists a positive-definite matrix $\mathcal{P}$ such that
	\begin{equation} \nonumber
	\begin{split}
	\mathcal{P}(A+\mathfrak{T}_1F_1\mathcal{K}_1)^T+(A+\mathfrak{T}_1F_1\mathcal{K}_1)\mathcal{P}<0
	\end{split}
	\end{equation} 
	for all admissible uncertainty $F_1(t)$ satisfying $F_1^TF_1\leq \varrho^2I$ if and only if there exists a scalar $q>0$ such that
	\begin{equation} \nonumber
	\begin{split}
	\mathcal{P}A^T+A\mathcal{P}+\frac{1}{q}\mathcal{P}\mathcal{K}_1^T\mathcal{K}_1\mathcal{P}+q\varrho^2\mathfrak{T}_1\mathfrak{T}_1^T<0.
	\end{split}
	\end{equation}
\end{lem}
Next, we present the main result on data-driven consensus with noise-corrupted data. 	
\begin{thm} \label{thm4}
	Assume that Assumptions \ref{asp 1}-\ref{asp 5} hold. The collected data $(X_i,U_{i-}), i=0,1,\cdots,N$ are informative for consensus, if there exists  $\Phi_i>0$, $F_i$ and scalars $\epsilon_i\geq0$, $\gamma_i>0$ and $\tau>0$ satisfying the following LMIs:

	\begin{equation} \label{equ:S-condition}
	\begin{split}
	&\begin{bmatrix}
	\Phi_i-\gamma_iI & 0 & 0 & 0 & 0 & 0\\
	0 & 0 & 0 & \Phi_i & 0 & 0\\
	0 & 0 & -\tau\nu^2 I &  F_i & 0 & 0\\
	0 & \Phi_i^T &  F_i^T & \Phi_i & F_i^T & 0\\
	0 & 0 & 0 & F_i & \tau I & 0 \\
	0 & 0 & 0 & 0 & 0 & I 
	\end{bmatrix} 
	-\epsilon_i \begin{bmatrix}
	I &  X_+\\
	0 &  -X_{i-}\\
	0 &  -U_{i-}\\
	0 & 0\\
	0 & 0\\
	0 & 0
	\end{bmatrix}\\
	&\times \begin{bmatrix}
	{N}_{11} & {N}_{12} \\
	{N}_{21} & {N}_{22}
	\end{bmatrix}\begin{bmatrix}
	*
	\end{bmatrix}^T>0, i=0,1,\cdots,N,
	\end{split}
	\end{equation}
	where $\nu=\frac{\lambda_N-\lambda_1}{\lambda_N+\lambda_1}$, 
	$\lambda_N$ and $\lambda_1$ denote the largest and smallest eigenvalues of $L_{ff}$, respectively. Then, the protocol (\ref{equ:control_law_noise}) with $\alpha=\frac{2}{\lambda_1+\lambda_N}$ and  $K_i(0)=F_i\Phi_i^{-1}$ { solves Problem 2. 	}
\end{thm}

\begin{proof}
Define $$\begin{aligned}
\Psi_i &=[I_N\otimes A+\alpha L_{ff}\otimes BK_i(0)]^T(I_N\otimes P_i)\\
&\quad\times[I_N\otimes A+\alpha L_{ff}\otimes BK_i(0)]-I_N\otimes P_i,
\end{aligned}$$
where $P_i>0$.
According to Definition 1,  the collected data $(X_i,U_{i-})$ are informative for consensus if there exists appropriate $K_i(0)$ and $P_i$ such that $\Psi_i<0$. Let $\Phi_i=P_i^{-1}$ and $F_i=K_i(0)\Phi_i$.  Consequently, $\Psi_i<0$ is equivalent to that 
$$\Phi_i-(A\Phi_i+\alpha \lambda_k BF_i)^T\Phi_i^{-1}(A\Phi_i+\alpha\lambda_k BF_i)>0,$$ for $k=1,\cdots,N$, where $\lambda_k$ denotes the $k$th eigenvalue of $L_{ff}$. This implies that $\Psi_i<0$ can be transformed into the above $N$ inequalities. Choose $\alpha=\frac{2}{\lambda_1+\lambda_N}$.
Evidently, $-\nu\leq\alpha \lambda_k-1\leq \nu$ for $k=1,\cdots,N$.
Next, motivated by \cite{Li2017Robust} and \cite{Li2012Distributed}, it can be inferred that  $A\Phi_i+\alpha \lambda_kBF_i$ is Schur stable for $k=1,\cdots,N$ if $A\Phi_i+(1+\Delta)BF_i$ is Schur stable for all $|\Delta|\leq \nu$. Then, it follows that $\Psi_i<0$, if there exists $\Phi_i>0$ such that  $$\Phi_i-(A\Phi_i+(1+\Delta) BF_i)\Phi_i^{-1}(A\Phi_i+(1+\Delta) BF_i)^T>0,$$ which is equivalent to
\begin{equation} \label{equ:S-lemma1}
\begin{split}
&\begin{bmatrix}
I\\
A^T \\
B^T
\end{bmatrix}^T \begin{bmatrix}
\Phi_i & 0\\
0 & -\begin{bmatrix}
\Phi_i\\
(1+\Delta) F_i 
\end{bmatrix}(\Phi_i)^{-1}\begin{bmatrix}
*
\end{bmatrix}^T
\end{bmatrix}\begin{bmatrix}
I\\
A^T \\
B^T
\end{bmatrix}>0.
\end{split}
\end{equation}
Note that
\begin{equation} \label{equ:SS1}
\begin{split}
\begin{bmatrix}
\Phi_i & 0\\
0 & -\begin{bmatrix}
\Phi_i\\
(1+\Delta) F_i 
\end{bmatrix}(\Phi_i)^{-1}\begin{bmatrix}
*
\end{bmatrix}^T
\end{bmatrix}>0,
\end{split}
\end{equation}
if and only if	\begin{equation} \label{equ:SS2}
\begin{split}
&\begin{bmatrix}
\Phi_i & 0 & 0 & 0 & 0\\
0 & 0 & 0 & \Phi_i & 0\\
0 & 0 & 0 & F_i & 0\\
0 & \Phi_i & F_i^T & \Phi_i & 0\\
0 & 0 & 0 & 0 & I
\end{bmatrix}
+\begin{bmatrix}
0\\
0\\
0\\
F_i^T\\
0
\end{bmatrix}\Delta\begin{bmatrix}
0 & 0 & I & 0 & 0
\end{bmatrix}\\
&\,\,+\begin{bmatrix}
0\\
0\\
I\\
0\\
0
\end{bmatrix}\Delta\begin{bmatrix}
0 & 0 & 0 & F_i & 0
\end{bmatrix}>0.
\end{split}
\end{equation}
{ 	Utilizing Lemma \ref{lem 4}, (\ref{equ:SS2}) holds for all $|\Delta|<\nu$ if and only if there exists a scalar $\tau>0$ such that}
\begin{equation}\label{equ:SS3}
\begin{split}
\begin{bmatrix}
\Phi_i & 0 & 0 & 0 & 0 & 0\\
0 & 0 & 0 & \Phi_i & 0 & 0\\
0 & 0 & -\tau\nu^2 I &  F_i & 0 & 0\\
0 & \Phi_i^T &  F_i^T & \Phi_i & F_i^T & 0\\
0 & 0 & 0 & F_i & \tau I & 0 \\
0 & 0 & 0 & 0 & 0 & I 
\end{bmatrix} >0.
\end{split}
\end{equation}
It is worth noting that using the Schur Complement lemma \cite{Boyd1994Linear} and pre- and post-multiplying  $\begin{bmatrix}
I & 
A &
B
\end{bmatrix}$ and $\begin{bmatrix}
I & 
A &
B
\end{bmatrix}^T$ on (\ref{equ:SS3}) directly leads to (\ref{equ:S-lemma1}), implying that (\ref{equ:SS3}) is a sufficient condition of (\ref{equ:S-lemma1}).

Note that all systems $(A,B)$ in $\mathcal{Q}_i$ satisfy the following constraint:
\begin{equation}\nonumber
\begin{split}
\begin{bmatrix}
I\\
A^T \\
B^T\\
\end{bmatrix}^T
\end{split}\begin{bmatrix}
I &  {X}_{i+}\\
0 &  -{X}_{i-}\\
0 &  -{U}_{i-}\\
\end{bmatrix} \begin{bmatrix}
{N}_{11} & {N}_{12} \\
{N}_{21} & {N}_{22}
\end{bmatrix}\begin{bmatrix}
*
\end{bmatrix}^T>0.
\end{equation}
Then, using the standard S-procedure in the matrix version \cite{Waarde2022} for (\ref{equ:noisy_system_identification}) and (\ref{equ:SS3}) directly leads  to (\ref{equ:S-condition}). 

Next, we can conclude that if (\ref{equ:S-condition}) holds, then (\ref{equ:S-lemma1}) holds for all $(A,B)$ in $\mathcal{Q}_i$. It then follows that $\Psi_i<0$,  implying that the data $(X_i,U_{i-})$ are informative for consensus and thereby $I_N\otimes A+\alpha L_{ff}\otimes BK_i(0)$ is Schur stable with $\alpha=\frac{2}{\lambda_1+\lambda_N}$ and  $K_i(0)=F_i\Phi_i^{-1}$.

Finally, we need to prove that the proposed control protocol (\ref{equ:control_law_noise}) along with the feedback gain matrix $K_i(0)$ obtained by (\ref{equ:S-condition}) can achieve consensus for the agents in (\ref{equ:1}).
Define $\tilde{x}_i(t)=x_i(t)-x_0(t)$ and $\tilde{x}(t)=\begin{bmatrix}
\tilde{x}_1^T(t) & \tilde{x}_2^T(t) & \cdots & \tilde{x}_N^T(t)
\end{bmatrix}^T$. Substituting (\ref{equ:control_law_noise}) into (\ref{equ:1}) gives
\begin{equation} \label{equ:closed_loop_noisy}
\begin{split}
\tilde{x}(t)=&[I_{N}\otimes A+L_{ff}\otimes BK_0(0)]\tilde{x}(t)\\
&+\begin{bmatrix}
(L_{ff})_1\otimes B\tilde{K}_1(t)\\
\vdots\\
(L_{ff})_{N}\otimes B\tilde{K}_N(t)
\end{bmatrix}\tilde{x}(t),
\end{split}
\end{equation}
where $(L_{ff})_i$ represents the $i$th row of $L_{ff}$ and $\tilde{K}_i(t)=K_i(t)-K_0(0)$.
It can be inferred from \cite{Liu2018} that $K_i(t)$ exponentially converges to $K_0(0)$ for $i=1,\cdots,N$. It is worth noting that such $K_0(0)$ renders $I_N\otimes A+\alpha L_{ff}\otimes BK_0(0)$ Schur stable for all $(A,B)$ in $\mathcal{Q}_0$,  as evidenced by the aforementioned analysis. Then we can derive from (\ref{equ:closed_loop_noisy}) that $\tilde{x}(t)\to 0$ in view of Lemma \ref{lem 3}. This completes the proof. \end{proof}

\begin{rem}
	
	In Theorem \ref{thm4}, we propose a new paradigm different from the methods in Section \ref{sec 3}, to obtain the initial data-based feedback gain matrix $K_i(0)$, leveraging the S-procedure and the informativity approach to account for the additive noise $d$ in (\ref{equ:noise_system}). The primary motivation stems from the inherent limitation of the conventional data-driven LQR algorithm since it cannot yield the accurate solution, which is essential in determining the consensus region, when dealing with noise-corrupted data.
\end{rem}

\begin{rem}
	The works in \cite{Wang2022} also consider noisy data-driven consensus control of multi-agent systems. Nevertheless, the network system in \cite{Wang2022} is transformed into a single linear system represented in compact forms, directly leading to a high-dimensional LMI  that is hard to solve for large-scale networks. On the contrary, the method given in Theorem \ref{thm4} allows each agent to compute its own initial gain matrix with a low-dimensional LMI, making the proposed algorithm more applicable and accessible in complex network scenarios.  
\end{rem}

\section{Localized Data-driven Control with Leader's Data}\label{sec Centralized}

Note that we sample data from all the followers, and from both the leader and the followers in Sections \ref{sec 3} and \ref{sec 4}, respectively, to construct the initial feedback gain matrices $K_i(0)$. In this section, we provide another data-driven consensus algorithm encompassing both noiseless and noisy scenarios, in which only the leader samples its own data and computes its feedback gain matrix while the followers obtain such gain information via the designed interaction scheme.  The system dynamics are still described by  (\ref{equ:1}), and the control input of the leader agent is set to $u_0=0$.

We start from the simple case where the sampled data are noise-free.  We only collect state and input data on $T$ finite sequences from the leader and construct the following data matrices:
\begin{equation} \label{equ:single data}
\begin{split}
U_{-}=&\begin{bmatrix}
u_0(0) & u_0(1) & \cdots & u_0(T-1)
\end{bmatrix}, \\
X=&\begin{bmatrix}
x_0(0) & x_0(1) & \cdots & x_0(T)
\end{bmatrix}.
\end{split}
\end{equation}
Define
\begin{equation} \nonumber
\begin{split}
X_{-}=&\begin{bmatrix}
x_0(0) & x_0(1) & \cdots & x_0(T-1) 
\end{bmatrix}, \\
X_{+}=&\begin{bmatrix}
x_0(1) & x_0(2) & \cdots & x_0(T)
\end{bmatrix}.
\end{split}
\end{equation}

\begin{assum} \label{asp 6}
	The communication graph contains a directed spanning tree with the leader as the root node.
\end{assum}

The control protocol for the agents in (\ref{equ:1}) is proposed as follows:
\begin{equation} \label{equ:control_law_3}
\begin{split}
u_i(t)&=c_i(t)K_i(t)\sum_{j=0}^{N}w_{ij}\bigg(x_i(t)-x_j(t)\bigg), \\
K_i(t+1)&=K_i(t)+\sum_{j=0}^{N}w_{ij}\bigg(K_j(t)-K_i(t)\bigg),\\
c_i(t+1)&=c_i(t)+\sum_{j=0}^{N}w_{ij}\bigg(c_j(t)-c_i(t)\bigg), 
\end{split}
\end{equation}
for $i=1,\cdots,N$,
where $c_0(0)=c_0$, $K_0(0)=K_0$, $K_0$ and $c_0$ are the data-based gains calculated by the leader, $w_{ij}$ and $w_{ii}$ are defined as in (\ref{equ:control_law_noise}).

It can be inferred from \cite{Liu2018} that, under Assumption \ref{asp 6}, the feedback gains $K_i(t)$ and $c_i(t)$ in (\ref{equ:control_law_3})
exponentially converge to $K_0$ and $c_0$, respectively. The problem we want to solve in this section is then described as follows:

\begin{prob}\label{prob 3}
	Design the gain matrix $K_0$ and the coupling gain $c_0$ using collected data (\ref{equ:single data}) such that the control law (\ref{equ:control_law_3}) achieves leader-follower consensus.
\end{prob}

In the following theorem, we present the main result of this section.
\begin{thm} \label{thm5}
	Suppose that Assumptions \ref{asp 3}, \ref{asp 6} hold and there exists $\Gamma$ optimizing (\ref{equ:K0}).
	Then the feedback gain matrix $K_0$ can be calculated as $K_0=U_-\Gamma(X_-\Gamma)^{-1}$. The solution $P$ to the ARE (\ref{equ:mare}) is obtained via (\ref{equ:P_noise}).
	Let $\mathcal{M}$ be one solution to the following equation:
	\begin{equation} \label{equ:L0}
	\begin{split}
	\begin{bmatrix}
	K_0\\
	0
	\end{bmatrix}=
	\begin{bmatrix}
	U_- \\
	X_-
	\end{bmatrix}\mathcal{M}.
	\end{split}
	\end{equation}
	Let $\theta=\sigma_{{\rm max}}(Q^{-1/2}\mathcal{M}^TX_+^TPX_+\mathcal{M}Q^{-1/2})$, where  $Q>0$ is a parameter in (\ref{equ:mare}).
	If there exists  a circle $C(h_0,r_0)$ covering all the eigenvalues of $\bar{L}$ such that \begin{equation} \nonumber
	\frac{r_0}{h_0}<\theta^{-1/2},
	\end{equation}
	then protocol (\ref{equ:control_law_3}) with $K_0=U_-\Gamma(X_-\Gamma)^{-1}$ and $c_0=\frac{1}{h_0}$ solves Problem \ref{prob 3}.
\end{thm}

\begin{proof}  
Let $\tilde{x}_i(t)=x_i(t)-x_0(t)$, $\tilde{K}_i(t)=K_i(t)-K_0$, and $\tilde{c}_i(t)=c_i(t)-c_0$ for $i=1,\cdots,N$. 
Then, substituting (\ref{equ:control_law_3}) into (\ref{equ:1}) gives
\begin{equation} \label{equ:closed_loop3}
\begin{split}
\tilde{x}_i(t+1)=&A\tilde{x}_i(t)+c_0BK_0\sum_{j=0}^{N}w_{ij}[\tilde{x}_i(t)-\tilde{x}_j(t)] \\
+&\tilde{c}_i(t)BK_0\sum_{j=0}^{N}w_{ij}[\tilde{x}_i(t)-\tilde{x}_j(t)]\\
+&c_0B\tilde{K}_i(t)\sum_{j=0}^{N}w_{ij}[\tilde{x}_i(t)-\tilde{x}_j(t)]\\
+&\tilde{c}_i(t)B\tilde{K}_i(t)\sum_{j=0}^{N}w_{ij}[\tilde{x}_i(t)-\tilde{x}_j(t)]. 
\end{split}
\end{equation}

Let $\tilde{x}(t)=[\tilde{x}_0^T(t),\cdots,\tilde{x}_N^T(t)]^T$ and $\tilde{J}_i(t)=\tilde{c}_i(t)K_0+c_0\tilde{K}_i(t)+\tilde{c}_i(t)\tilde{K}_i(t)$. Then,
following similar lines in the proof of Theorem \ref{thm3}, we rewrite (\ref{equ:closed_loop3}) into the following compact form:
\begin{equation} \label{equ:leader_closed_loop}
\begin{split}
\tilde{x}(t+1)=&[I_{N}\otimes A+c_0\bar{L}\otimes BK_0]\tilde{x}(t)\\
&+\begin{bmatrix}
\bar{L}_1\otimes B\tilde{J}_1(t)\\
\vdots\\
\bar{L}_{N}\otimes B\tilde{J}_N(t)
\end{bmatrix} \tilde{x}(t),
\end{split}
\end{equation}
where $\bar{L}_i$ represents the $i$th row of $\bar{L}$. 
It can be inferred from \cite{Liu2018} that $\tilde{J}_i(t)\to 0$ exponentially fast. Therefore, to prove Theorem \ref{thm5}, it remains to ensure that $I_{N+1}\otimes A+c_0\bar{L}\otimes BK_0$ is Schur stable according to Lemma \ref{lem 3}.

Define 
$\theta=\sigma_{\rm {max}}(Q^{-\frac{1}{2}}{A}^TP{B}({B}^TP{B})^{-1}{B}^TP{A}Q^{-\frac{1}{2}})$,
where $P$ is the solution to the ARE (\ref{equ:mare}).
Recalling the model-based consensus algorithm in \cite{Movric2013Synchronization},  if there exists  a circle $C(h_0,r_0)$ covering all the eigenvalues of $\bar{L}$ such that $r_0/h_0<\theta^{-1/2}$, then the matrix $I_{N}\otimes A+c_0\bar{L}\otimes BK_0$ is Schur stable with $c_0=1/h_0$ and $K_0=-(B^TPB)^{-1}B^TPA$. 

Since $(A,B)$ is unknown, we cannot obtain such $K_0$, $P$, and $\theta$ directly. Therefore, we design the feedback gain matrix $K_0$ as $K_0=U_-\Gamma(X_-\Gamma)^{-1}$, where $\Gamma$ optimizes (\ref{equ:K0}). It can be inferred from Theorem \ref{thm1} that such $K_0$ satisfies  $K_0=-(B^TPB)^{-1}B^TPA$. The solution $P$ to (\ref{equ:mare}) is then obtained via (\ref{equ:P_noise}). Note that $\theta$ can be rewritten as $\theta=\sigma_{\rm {max}}(Q^{-\frac{1}{2}}K_0^TB^TPBK_0Q^{-\frac{1}{2}})$. Suppose $\mathcal{M}$ is a solution to (\ref{equ:L0}). It then follows that $BK_0=X_+\mathcal{M}$. Consequently, $\theta$ can be represented by $\theta=\sigma_{{\rm max}}(Q^{-1/2}\mathcal{M}^TX_+^TPX_+\mathcal{M}Q^{-1/2})$. We can draw a conclusion that if $K_0$ and $c_0$ satisfy the condition in Theorem \ref{thm5}, then $I_{N}\otimes A+c_0\bar{L}\otimes BK_0$ is Schur stable. In view of Lemma \ref{lem 3}, we can deduce from (\ref{equ:leader_closed_loop}) that $\tilde{x}(t)\to0$. This completes the proof. 
\end{proof}

The case where the collected data are noise-corrupted can be similarly studied, following similar steps in Theorems \ref{thm4} and \ref{thm5}. The details are omitted here for brevity.

\begin{rem}
	In this section, we propose a different data-driven control architecture for the multi-agent systems, where the feedback gains are computed using data only sampled from the leader agent and transmitted to followers through an interaction mechanism. It is worth noting that the proposed control protocol (\ref{equ:control_law_3})  is still devised in a distributed fashion without requiring a centralized node.
\end{rem}

\section{Simulation Results} \label{sec 5}
In this section, we will use a simulation example to verify the effectiveness of the proposed schemes. 
The dynamics of the discrete-time agents
are given by (\ref{equ:1}), with
\begin{equation}\nonumber
\begin{split}
A=\begin{bmatrix}
0.707 & 0.707 \\
-0.707 & 0.707
\end{bmatrix}, 
B=\begin{bmatrix}
0.2 & 0\\
0 & 0.2
\end{bmatrix}.
\end{split}
\end{equation}
The network topology is {chosen to be} in Fig. 1, with agent 0 being the leader. 
The subgraph among followers is undirected.
The submatrix $L_{ff}$  { is given by}

\begin{equation} \nonumber
L_{ff}=\begin{bmatrix}
3 & -1 & -2 & 0 & 0 \\
-1 & 10 & -1 & -3 & 0 \\
-2 & -1 & 10 & 0 & -2 \\
0 & -3 & 0 & 10 & -2\\
0 & 0 & -2 & -2 & 4
\end{bmatrix}. 
\end{equation}

{It should be noted that the edge weights between the leader and the followers connected to it are relatively large to satisfy the condition in Lemma \ref{lem 2}, which can be referred to \cite{Movric2013Synchronization} for more details.}

\begin{figure}
	\begin{center}
		\includegraphics[width=6.0cm]{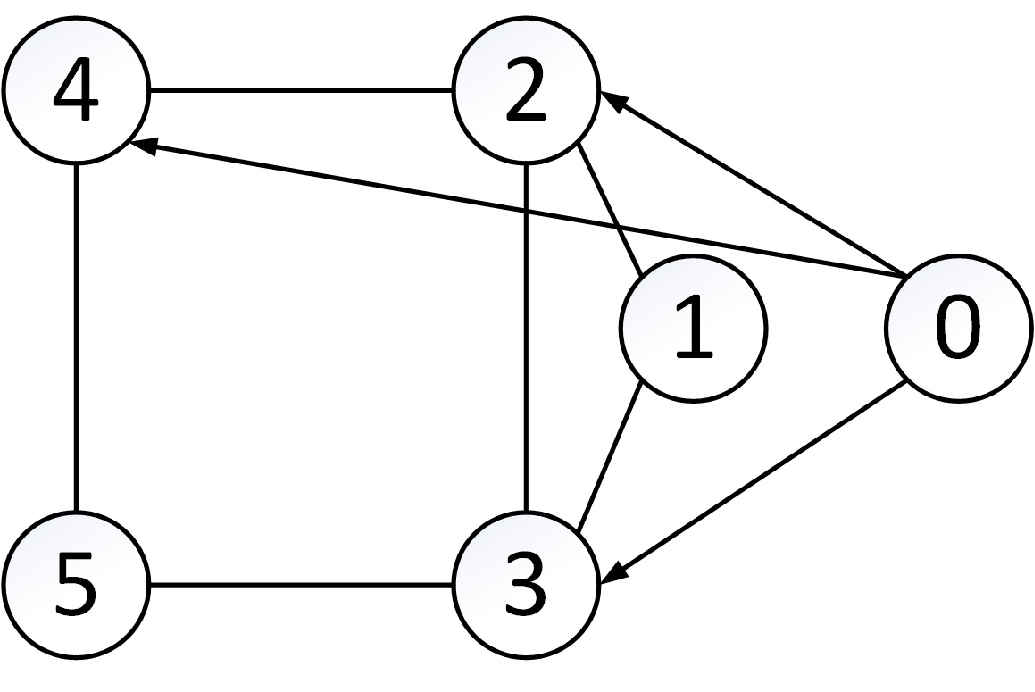}    
		\caption{The communication topology. }
		\label{fig:topology}
	\end{center}
\end{figure}

\begin{figure}
	\begin{center}
		\includegraphics[width=8.0cm]{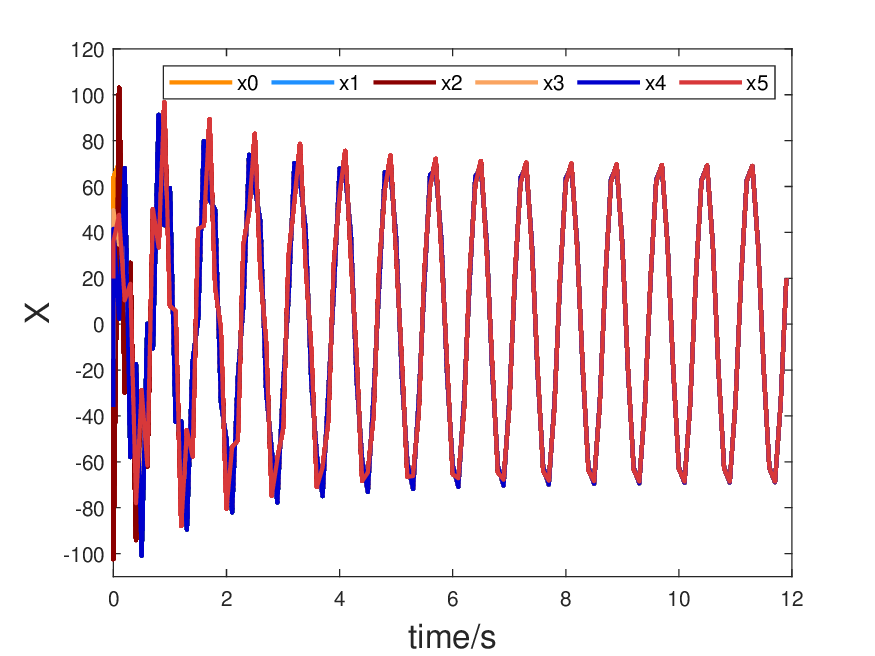}    
		\caption{Trajectories of agents on the X-axis for the case of noiseless data. }
		\label{fig:consensus_error_real}
	\end{center}
\end{figure}

\begin{figure}
	\begin{center}
		\includegraphics[width=8.0cm]{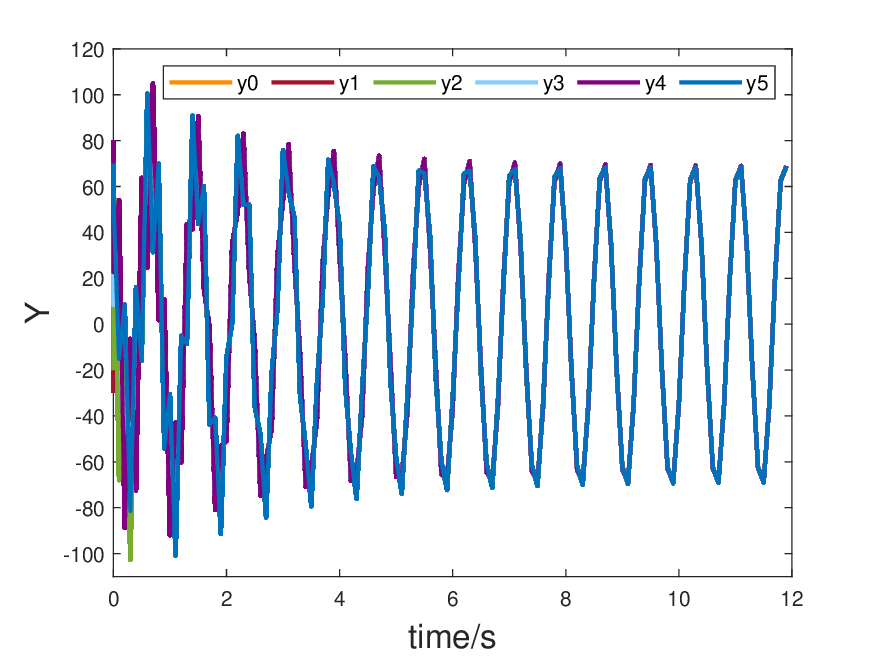}    
		\caption{Trajectories of agents on the Y-axis for the case of noiseless data.}
	\end{center}
\end{figure}

In the case of noiseless data, we generate data with random initial conditions and by applying to each input channel a random sequence. Next,  we utilize CVX \cite{CVX} to solve (\ref{equ:K0}), yielding initial data-based $K_i(0)$ as follows:
\begin{equation}\nonumber
\begin{split}
K_{1}(0)=\begin{bmatrix}
-5.2701 & -4.9733\\
6.5984 & -4.2020
\end{bmatrix}, \\
K_{2}(0)=\begin{bmatrix}
-6.7275 & -5.7100\\
6.6883 & -3.8911
\end{bmatrix}, \\
K_{3}(0)=\begin{bmatrix}
-3.9119 & -3.7200\\
3.5389 & -3.8300
\end{bmatrix}, \\
K_{4}(0)=\begin{bmatrix}
-5.8621 & -4.6591\\
4.1762 & -3.3509
\end{bmatrix}, \\
K_{5}(0)=\begin{bmatrix}
-3.6962 & -3.3449\\
4.0185 & -3.9898
\end{bmatrix}.
\end{split}
\end{equation}
It then follows that the upper bound of the consensus region in Theorem \ref{thm3} is calculated as $\frac{1}{\sigma_{\rm max}(\mathcal{F}^{-1/2}\mathcal{R}\mathcal{F}^{-1/2})}=0.8519$. Note that all  eigenvalues of $\bar{L}=(I+\mathcal{D}_{ff})^{-1}L_{ff}$ lie within the graph circle $\bar{C}(1,0.3)$, which satisfies the constraint  in Theorem \ref{thm3}.
By substituting the calculated $K_i(0)$ into (\ref{equ:1}),  the state trajectories of the controlled agents are demonstrated in Fig. 2 and Fig. 3. As depicted in these figures, the proposed protocols (\ref{equ:dis_control_law}) successfully achieve consensus with noiseless data. 

For the case of noisy data, we add energy-bounded noises, drawn randomly from a Gaussian distribution with zero mean and unit variance, to the measurements of the agents' dynamics. The noise signals, denoted as $D_i$, adhere to the constraint in (\ref{equ:noise_constraint}), where $N_{11}=0.1I$, $N_{22}=-I$, and $N_{12}=N_{21}=0$. Solving (\ref{equ:S-condition}) via CVX also yields distinct initial feedback gain matrices $K_i(0)$ as below:
\begin{equation}\nonumber
\begin{split}
K_{0}(0)=\begin{bmatrix}
-0.1248 & -0.0460\\
0.0005 & -0.2301
\end{bmatrix}, \\
K_{1}(0)=\begin{bmatrix}
-0.1212 & -0.0401\\
0.0052 & -0.2271
\end{bmatrix},\\
K_{2}(0)=\begin{bmatrix}
-0.1246 & -0.0236\\
0.0028 & -0.2261
\end{bmatrix}, \\
K_{3}(0)=\begin{bmatrix}
-0.1241 & -0.0336\\
-0.0026 & -0.2292
\end{bmatrix},\\
K_{4}(0)=\begin{bmatrix}
-0.1250 & -0.0194\\
0.0071 & -0.2237
\end{bmatrix},\\
K_{5}(0)=\begin{bmatrix}
-0.1213 & -0.0140\\
0.0065 & -0.2132
\end{bmatrix}.
\end{split}
\end{equation}
Subsequently, we obtain the state trajectories of agents as shown in Fig. 4 and Fig. 5. It is manifest from these figures that the agents in (\ref{equ:1})   reach consensus under the proposed protocols (\ref{equ:control_law_noise}) in the presence of noise-corrupted data.
\begin{figure}
	\begin{center}
		\includegraphics[width=8.0cm]{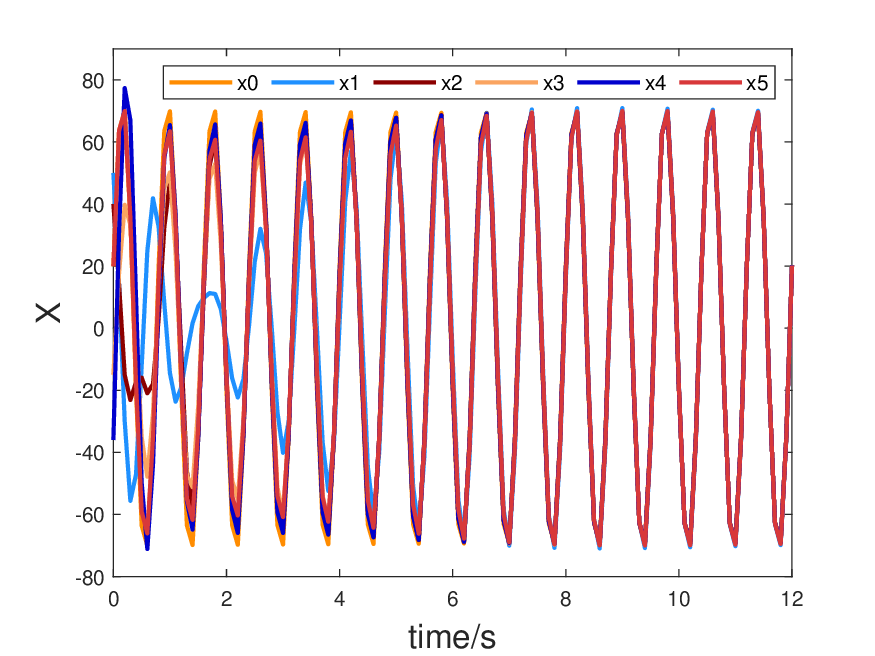}    
		\caption{Trajectories of agents on the X-axis for the case of noisy data.}
	\end{center}
\end{figure}

\begin{figure}
	\begin{center}
		\includegraphics[width=8.0cm]{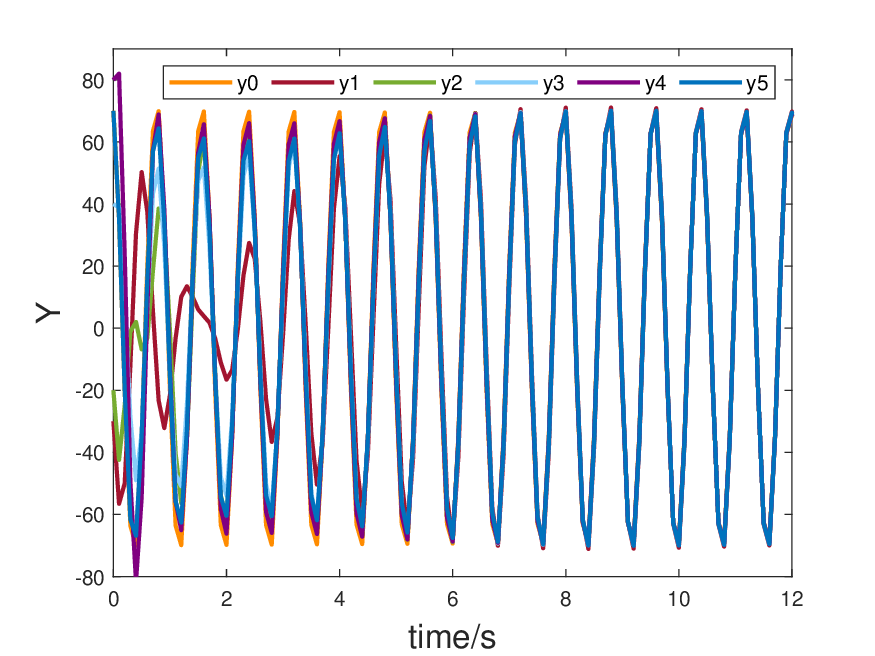}    
		\caption{Trajectories of agents on the Y-axis for the case of noisy data.}
	\end{center}
\end{figure}

\section{Conclusion} \label{sec 6}

In this paper, we have proposed and studied the localized data-driven consensus control problem for leader-follower multi-agent systems, allowing each agent to compute its local control gain with its locally collected data. Both the noiseless and noisy data-driven consensus control problems are addressed by solving low-dimensional LMIs. We have also extended the results to the case where only the leader's data are sampled and utilized.
Potential future research includes data-driven consensus control with output-feedback design, synchronization control for continuous-time multi-agent systems, and the integration of event-triggered mechanisms utilizing sampled data.

\section*{References}
\bibliographystyle{IEEETran}
\bibliography{reference}

\end{document}